\documentclass[11pt]{article}
\usepackage{amsfonts,latexsym}
\usepackage{amsthm}
\usepackage{amsmath}
\usepackage[dvips]{graphicx}
\usepackage{graphicx}
\usepackage{amscd}

\usepackage{amssymb}
\usepackage{a4wide}

\newtheorem{theorem}{Theorem}

\numberwithin{equation}{section}
\newcommand{\ohf}{\frac{1}{2}}

\newcommand{\Sfrac}[2]{{ \textstyle \frac{#1}{#2}}}

\newcommand{\R}{{\mathbb{R}}}
\newcommand{\scrO}{{\mathcal{O}}}

\newcommand{\ve}{{\varepsilon}}

\newcommand{\dtt}{\frac{\partial}{\partial t}}
\newcommand{\dxx}{\frac{\partial}{\partial x}}

\newcommand{\dx}{\partial_x}

\newtheorem{remark}{Remark}

\title{Approximate Conservation Laws in the KdV Equation}

\author{Samer Israwi\footnote{\texttt{s\_israwi83@hotmail.com}} \\
{\small Lebanese University, Faculty of Sciences 1} \\
{\small Hadath-Beirut, Lebanon} \\ \\
Henrik Kalisch\footnote{\texttt{henrik.kalisch@uib.no}} \\
{\small Department of Mathematics, University of Bergen} \\
{\small Postbox 7800, 5020 Bergen, Norway} \\
}

\date{}

\begin{document}
\maketitle
%
%
%
\begin{abstract}
The Korteweg-de Vries equation is known to yield a valid description 
of surface waves for waves of small amplitude and large wavelength.
The equation features a number of conserved integrals, 
but there is no consensus among scientists as to the physical meaning of these integrals. 
In particular, it is not clear whether these integrals are related to the conservation
of momentum or energy, and some researchers have questioned the conservation of energy
in the dynamics governed by the equation.
In this note it is shown that while exact energy conservation may not hold,
if momentum and energy densities and fluxes are defined in an appropriate way, then solutions 
of the Korteweg-de Vries equation give rise to {\it approximate}
differential balance laws for momentum and energy.
\end{abstract}
\section{Introduction}
The Korteweg-de Vries (KdV) equation 
\begin{align}\label{KdV_nondim}
\eta_{t}+ \eta_{x} + \ve \frac{3}{2} \eta \eta_{x}
                    + \ve \frac{1}{6} \eta_{xxx} = 0
\end{align}
is one of the most widely studied equations 
in mathematical physics today, and it stands as a paradigm in the field of 
completely integrable partial differential equations. 
The KdV equation admits a large number of closed-form solutions
such as the solitary wave, the cnoidal periodic solutions,
multisolitons and rational solutions.
It also features an infinite number of formally conserved
integrals which is one of the hallmarks of a completely integrable system.
Indeed the conservation can be made mathematically rigorous 
using the techniques developed in \cite{BonaSmith}.

While our understanding of this model equation is generally rather complete,
there appears to be one aspect which has not received much attention.
Indeed it seems that the link between the invariant integrals of the equation
and physical conservation laws has not been well understood.
In the present note we explore the ramifications of imposing
mechanical balance laws such as momentum and energy conservation
in the context of the KdV equation.

To explain this point further, recall that if the equation 
is given in the form \eqref{KdV_nondim}
then the first three conserved integrals are
\begin{align}
\label{conslaws}
\int_{- \infty}^{\infty} \eta \, dx,
\qquad
\int_{- \infty}^{\infty} \eta^2 \, dx,
\quad \mbox{ and } \quad
\int_{- \infty}^{\infty} \left( \Sfrac{1}{3} \eta_x^2  - \eta^3 \right) \, dx.
\end{align}
The first integral is found to be invariant with respect to time $t$
as soon as it is recognized that the KdV equation can be written in
the form
\begin{align}\label{kdv_mass}
\dtt \left( \eta \right) + 
\dxx \Big( \eta + \ve \frac{3}{4} \eta^2 + \ve \frac{1}{6} \eta_{xx} \Big) = 0.
\end{align}
The quantity appearing under the time derivative can clearly be
interpreted as the excess mass density, and as shown in \cite{AK4}, 
the expressions appearing under the spatial
derivative is the mass flux through a cross section of unit width
due to the passage of a surface wave.

Invariance of the second and third integrals is obtained from the identities
\begin{align}\label{kdv_II}
\dtt \left( \eta^2 \right) 
+ \dxx \left( \eta^2 +  \ve \eta^3 + \frac{\ve}{3} \eta \eta_{xx} 
- \frac{\ve}{6} \eta_{x}^2 \right) = 0,
\end{align}
\begin{align}\label{kdv_III}
\dtt \left( \eta^3 - \frac{1}{3} \eta_x^2 \right) 
+ \dxx \left( \eta^3 + \ve \frac{9}{8}  \eta^4 + \frac{2}{3} \eta_x \eta_t 
\right.
\left.
+ \frac{1}{3} \eta_{x}^2 +  \ve \frac{1}{18} \eta_{xx}^2 + \ve \frac{1}{2} \eta^2 \eta_{xx} \right) = 0.
\end{align}
When contemplating these formulas, the question naturally arises
if the invariant integrals in \eqref{conslaws} also have a physical meaning. 
Of course, the first integral represents the total excess mass,
but we do not know whether the second or third integral represent
any pertinent mechanical quantity.
A second natural question is whether the fluxes appearing in \eqref{kdv_II}
and \eqref{kdv_III} have a physical meaning.

As a matter if fact, it was noted by the authors of \cite{AS} that the densities and fluxes 
appearing in \eqref{kdv_II} and \eqref{kdv_III} do not represent any concrete physical quantities.
More recent work also clearly casts doubt on the exact conservation 
of energy in the KdV equation \cite{KRI2015,KRIR2017}.
In light of these findings, one may then ask what the appropriate
densities and fluxes are if momentum and energy conservation are
to be understood in the context of the KdV equation.

The main result of the present letter is that one may define
densities and fluxes that represent momentum and energy
conservation and lead to {\it approximate}, but not
exact conservation.
Indeed, we will show that if these quantities are chosen
correctly, then momentum and energy conservation hold
to the same order as the KdV equation is valid. 
To be more precise, if we call the non-dimensional momentum density by 
$I(\eta)$, and the non-dimensional flow force 
(momentum flux plus pressure force) by $q_I(\eta)$, 
we obtain the approximate local balance law 
\begin{align}\label{intro_I}
\frac{\partial}{\partial t} I(\eta) + \frac{\partial}{\partial x} q_{I}(\eta)
                     = \scrO(\ve^2).
\end{align}
Similarly, the approximate energy balance 
\begin{align}\label{intro_E}
\frac{\partial}{\partial t} E(\eta) + \frac{\partial}{\partial x} q_{E}(\eta)
                     = \scrO(\ve^2)
\end{align}
follows if the expressions for the energy density $ E(\eta)$ 
and the energy flux plus work rate due to pressure force $ q_E(\eta)$
are chosen appropriately.

The plan of the paper is as follows.
In Section 2, we explain the physical background against which
the KdV equation is used as an approximate water-wave model.
The developments in Section 2 are based on firm mathematical
theory which has been developed in the last two decades,
and is summarized handily in \cite{LannesBOOK}.
Then using this background material, 
the statements of momentum and energy conservation introduced above
will be made mathematically precise in sections 3 and 4.

\section{The KdV equation in the context of surface water waves}

We study the KdV equation as a model equation for waves at the free surface of an incompressible, inviscid fluid 
running in a narrow open channel where transverse effects can be neglected. 
Let $h_0$ be the depth of the undisturbed fluid.
Denoting by $\lambda$ a typical wavelength and by $a$ a typical amplitude
of a wavefield to be described, the number
$\ve = a/h_0$ represents the relative amplitude, and 
$\mu = h_0^2 / \lambda^2$ measures the inverse relative wavenumber.
The geometric setup of the problem is indicated in Figure \ref{Figure1}.

In suitably non-dimensionalized variables,
the surface wave excursion of right-going waves can be shown to satisfy the relation
\begin{align*}
\eta_{t}+ \eta_{x} + \ve \frac{3}{2} \eta \eta_{x}
                    + \mu \frac{1}{6} \eta_{xxx} = \scrO(\ve^2, \ve \mu, \mu^2).
\end{align*}
If the wave motion is such that both $\ve$ and $\mu$ are small and of similar size,
then we can take equation \eqref{KdV_nondim}
to obtain an approximate description of the dynamics of the free surface.
The approximation can be made rigorous using the
techniques in \cite{BCL,Craig,Israwi,LannesBOOK,SchneiderWayne} and others.
Sometimes the Stokes number $\mathcal{S} = \ve / \mu$ is introduced
in order to quantify the applicability of the equation to a particular regime of surface waves.
Let us assume for the time being that the Stokes number is equal to unity,
so that we can work with a single small parameter $\ve$.
In this scaling, we may also assume that initial data $\eta_0$ are given,
such that for any $k>0$, we have $\|\dx^k \eta_0\|_{L^2} \le \scrO(1)$.

\begin{figure}
  \begin{center}
  \includegraphics[scale=0.8]{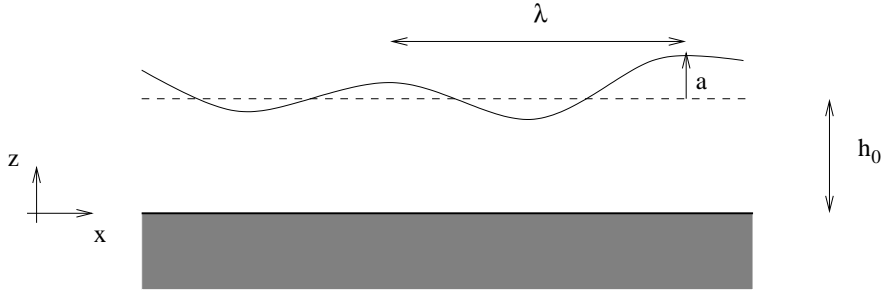}
  \end{center}
  \caption{\small The schematic elucidates the geometric setup of the problem.
           The free surface is described by a function $\eta(x,t)$, and
           the $x$-axis is aligned with the flat bed.
}
\label{Figure1}
\end{figure}

Using the aforementioned techniques, it can be shown that the 
velocity field $(\phi_{x},\phi_{z})$ 
at any non-dimensional height $z$ in the fluid column
can be expressed in terms of a solution of the KdV equation 
to second order accuracy in $\ve$ by
\begin{align}\label{phi_eq_x}
\phi_{x}(x,z,t) = \eta - \ve \frac{1}{4}\eta^2
                       + \ve \Big( \frac{1}{3}- \frac{z^2}{2} \Big) \eta_{xx},
\end{align}
\begin{align}\label{phi_eq_y}
\phi_{z}(x,z,t) = -  \ve z \eta_{x}.
\end{align}
%
%
%
%
Similarly, the pressure can be expressed in terms of a
solution $\eta$ of the KdV equation as follows.
First define the dynamic pressure by subtracting the
hydrostatic contribution:
\begin{equation*}
P' = P - P_{atm}+  g z.
\end{equation*}
Since the atmospheric pressure is of a magnitude much smaller than
the significant terms in the equation, it will be assumed to be zero.
As shown in \cite{AK2}, the dynamic pressure $P'$ can the be approximated 
to second order in $\ve$ by
\begin{equation}\label{P_expression}
P'= \eta - \ohf \ve (z^2-1)\eta_{xx}.
\end{equation}
%

\section{Approximate momentum balance}
%
%
\begin{figure}
  \begin{center}
  \includegraphics[scale=0.8]{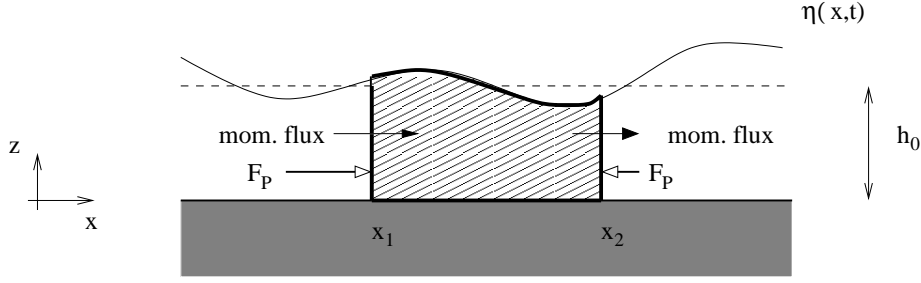}
  \end{center}
  \caption{\small The flow force $q_I$ represents the sum
                  of momentum flux and pressure force on
                  a fluid element of unit width, reaching from the free surface 
                  to the bed.}
\label{Figure2}
\end{figure}
Using ideas from \cite{AK1,AK2}, we can define an approximate momentum density
$I$ by substituting the horizontal velocity $\phi_x$ given by
\eqref{phi_eq_x} into the integral
$$
\int_{0}^{1+\ve\eta(x,t)} \phi_x(x,z,t) \, dz.
$$
In the same vein, the momentum flux and pressure force are 
constructed using the integrals
\begin{equation}
\int_{0}^{1+\ve \eta} \phi_x^2 \, dz  \quad \mbox{ and } \quad
\int_{0}^{1+\ve \eta} p \,  dz.
\end{equation}
In this way, using the asymptotic analysis delineated in \cite{AK4},
the horizontal momentum density in the KdV context is found to be
\begin{equation}\label{momDensity}
I = \ve \eta + \ve^2 \Sfrac{3}{4} \eta^2 + \ve^2 \Sfrac{1}{6} \eta_{xx},
\end{equation}
and the horizontal momentum flux plus pressure force is
\begin{equation}\label{momFlux}
q_I =  \Sfrac{1}{2} + \ve \eta + \ve^2 \Sfrac{3}{2} \eta^2 + \ve^2 \Sfrac{1}{3} \eta_{xx}.
\end{equation} 
The approximate local momentum balance can be formulated as follows.
\begin{theorem}
$\mathrm{(Momentum \, balance) }$
Suppose $\eta$ is a solution of \eqref{KdV_nondim} with initial data
$\eta_0$ satisfying $\| \partial_x^k \eta_0 \|_{L^2} = \scrO(1)$
for some integer $k \ge 5$.
Then there is a constant $C$, so that the estimate
$$
\left\| \dtt \Big\{ \eta + \ve \Sfrac{3}{4} \eta^2 + \ve \Sfrac{1}{6} \eta_{xx} \Big\}
  + \dxx \Big\{ \Sfrac{1}{2} + \eta + \ve \Sfrac{3}{2}\eta^2 + \ve \Sfrac{1}{3}\eta_{xx} \Big\} \right\|_{L^2} 
   \le C \ve^2
$$
holds for all $t \in [0,\infty)$.
\end{theorem}
\begin{proof}
\noindent
It was shown in Prop. 6 in \cite{BonaSmith} that given initial data
$\eta_0 \in H^k$, there exists a solution $\eta(x,t)$
which is bounded in the space $C(0,\infty,H^k)$.
Thus all ensuing computations hold rigorously for members of $H^k$.
Using the assumption that $\eta$ satisfies the KdV equation
and dividing through by $\ve$,
we can write the integrand in the statement of the theorem in the following way.
\begin{multline*}
\Big\{ \eta +  \ve \Sfrac{3}{4} \eta^2 +  \ve \Sfrac{1}{6} \eta_{xx} \Big\}_t 
+ \Big\{ \Sfrac{1}{2} + \eta + \ve \Sfrac{3}{2} \eta^2 + \ve \Sfrac{1}{3} \eta_{xx} \Big\}_x \\
=  \Big( \eta_t + \eta_x +  \ve \Sfrac{3}{2} \eta \eta_x +  \ve \Sfrac{1}{6} \eta_{xxx} \Big)
  + \ve \Big( \Sfrac{3}{4} \eta^2 + \Sfrac{1}{6} \eta_{xx} \Big)_t + \ve 
      \Big(\Sfrac{3}{4}\eta^2 + \Sfrac{1}{6}\eta_{xx} \Big)_x \\
= 0 \, + \,  \ve \Sfrac{3}{4} \Big( \eta_t^2 + \eta_x^2 \Big) +  \ve \Sfrac{1}{6} \Big(\eta_{xxt} + \eta_{xxx} \Big)
\qquad \qquad \qquad \qquad \qquad
\qquad  \\
=  \ve \Sfrac{3}{2} \eta \Big( - \ve \Sfrac{3}{2} \eta \eta_{x} - \ve \Sfrac{1}{6} \eta_{xxx} \Big) 
   + \ve \Sfrac{1}{6} \partial_x^2 \Big( - \ve \Sfrac{3}{2} \eta \eta_{x} - \ve \Sfrac{1}{6} \eta_{xxx} \Big) .
\qquad  \qquad 
\end{multline*}
Since $\|\eta_0\|_{H^5}$ is on the order of $1$, 
and $\|\eta(\cdot,t)\|_{H^k}$ is bounded for all time, we have
\begin{equation*}
\sup_t \left\| \dtt\left\{ \eta + \ve \Sfrac{3}{4} \eta^2 + \ve \Sfrac{1}{6} \eta_{xx} \right\}
  + \dxx \left\{\Sfrac{1}{2} + \eta + \ve \Sfrac{3}{2} \eta^2 + \ve \Sfrac{1}{3}\eta_{xx} \right\} \right\|_{L^2} 
\le C^2 \ve^2
\end{equation*}
for some constant $C$ for $t \in [0,\infty)$. \\
\end{proof}
\begin{remark}
It should be noted that the $L^2$ norm used in the proof can be replaced
by any Sobolev norm $H^k$ so long as the initial data are regular enough.
Indeed it can be seen immediately from the proof that if the approximate momentum 
balance is to be proved in $H^k(\R)$, then the initial data should
be given in $H^{k+5}(\R)$.
\end{remark}

Another important point is that the estimate in Theorem 1 is independent
of the time $t$. In other words, the momentum balance is global in the
sense that it holds as long as the solution of the KdV equation exists.
This is in stark contrast to the proofs providing the approximation
property of the KdV equation for the water-wave problem which are 
generally such that the error is bounded by $C\ve^2 (1+t)$,
so that as $t$ gets larger, the approximation degenerates, and
if $t \sim 1 / \ve$, the approximation is only on the order of $\ve$.
Thus in this sense the momentum balance is self-consistent
i.e. it is independent of the approximate nature of the KdV equation 
with regards to the water-wave problem.

\section{Approximate energy balance}
%
%
Using ideas from \cite{AK2,AK4}, we can define an approximate energy density
$E$ by substituting the horizontal velocity $\phi_x$ given by
\eqref{phi_eq_x} into the integral
$$
\int_{0}^{1+\ve\eta} \big\{ \phi_x^2 + \phi_z^2 +  z \big\} \, dz,
$$
where the last term in the integrand represent the potential energy.
In the same vein, the energy flux and work rate due to pressure force are 
constructed using the integrals
\begin{equation}
\int_{-0}^{1+\ve \eta} \big\{ \phi_x^2 + \phi_z^2 +  z  \big\} \phi_x \, dz  
\quad \mbox{ and } \quad
\int_{-0}^{1+\ve \eta} p \phi_x \,  dz.
\end{equation}
In this way, using again the asymptotic analysis explained in \cite{AK4},
the horizontal energy density in the KdV context is found to be
\begin{equation}
E = \frac{1}{2} + \ve \eta + \ve^2 \eta^2,
\end{equation}
and the horizontal energy flux plus work rate due to pressure force is
\begin{equation}
q_E =  \ve \eta + \ve^2 \frac{7}{4}\eta^2 + \ve^2 \frac{1}{6}\eta_{xx}.
\end{equation} 
For the energy balance, we have the following theorem. 
\begin{theorem}
$\mathrm{(Energy \, balance) }$
Suppose $\eta$ is a solution of \eqref{KdV_nondim} with initial data
$\eta_0$ satisfying $\| \partial_x^k \eta_0 \|_{L^2} = \scrO(1)$ for an integer $k \ge 4$. 
Then there is a constant $C$, so that the estimate
$$
\left\| \dtt \Big\{\Sfrac{1}{2} + \eta + \ve \eta^2 \Big\} 
     + \dxx \Big\{\eta + \ve \Sfrac{7}{4}\eta^2 + \ve \Sfrac{1}{6}\eta_{xx} \Big\} \right\|_{L^2} 
   \le \scrO(\ve^2)
$$
holds for all $t \in [0,\infty)$. 
\end{theorem}
\begin{proof}
The proof is similar, but simpler.
Observe that
\begin{multline}
 \dtt \Big\{\Sfrac{1}{2} + \ve \eta + \ve^2 \eta^2 \Big\} 
     + \dxx \Big\{\ve \eta + \ve^2 \Sfrac{7}{4}\eta^2 + \ve^2 \Sfrac{1}{6}\eta_{xx} \Big\} \\
= \ve \eta_t + \ve \eta_x + \ve^2 \frac{3}{2} \eta \eta_x + \ve^2 \frac{1}{6} \eta_{xxx}  + 2 \ve^2 \eta \eta_t + 2 \ve^2 \eta \eta_x \\
= - 3 \ve^3 \eta^2 \eta_x - \ve^3 \frac{1}{3} \eta \eta_{xxx}. \qquad \qquad
\end{multline}
The estimate now follows in the same way as for the momentum balance above.
\end{proof}

In some cases, it is convenient to normalize the potential energy differently,
so that the undisturbed state has zero potential energy. 
The approximate energy density $E^*$ is then defined 
by substituting the horizontal velocity $\phi_x$ given by
\eqref{phi_eq_x} into the integral
$$
\int_{0}^{1+\ve\eta} \big\{ \phi_x^2 + \phi_z^2 \big\} \, dz 
\ + \
\int_{1}^{1+\ve\eta}  z  \, dz.
$$
As shown in \cite{AK4}, in this case, the
energy density and energy flux have the respective form
$$
E^* = \ve^2 \eta^2 + \frac{1}{4} \ve^3 \eta^3 + \frac{1}{6} \ve^3 \eta \eta_{xx} 
+ \frac{1}{6} \ve^3 \eta_{x}^2,
$$ 
$$
q_E^* = \ve^2 \eta^2 + \frac{5}{4} \ve^3 \eta^3 + \frac{1}{2} \ve^3 \eta \eta_{xx}.
$$ 

\begin{theorem}
$\mathrm{(Energy \, balance) }$
Suppose $\eta$ is a solution of \eqref{KdV_nondim} with initial data
$\eta_0$ satisfying $\| \partial_x^k \eta_0 \|_{L^2} = \scrO(1)$
for some integer $k \ge 6$.
Then there is a constant $C$ such that the estimate
$$
\left\| \dtt \Big\{  \eta^2 + \frac{1}{4} \ve \eta^3 + \frac{1}{6} \ve \eta \eta_{xx} 
+ \frac{1}{6} \ve \eta_{x}^2 \Big\}
+ \dxx \Big\{  \eta^2 + \frac{5}{4} \ve \eta^3 + \frac{1}{2} \ve \eta \eta_{xx} \Big\} \right\|_{L^2}
\le C \ve^2
$$
holds for all $t \in [0,\infty)$.
\end{theorem}
\begin{proof}
The key computation is the following.
\begin{multline*}
 \dtt \Big\{  \eta^2 + \ve \frac{1}{4} \eta^3 + \ve \frac{1}{6} \eta \eta_{xx} 
                                             + \frac{1}{6} \ve \eta_{x}^2 \Big\}
  + \dxx \Big\{  \eta^2 + \ve \frac{5}{4} \eta^3 + \ve \frac{1}{2} \eta \eta_{xx} \Big\} \\
= 2 \eta (\eta_t + \eta_x) + \ve \frac{3}{4} \eta^2 \eta_t 
                          + \ve \frac{1}{6} \eta_t \eta_{xx} + \ve \frac{1}{6} \eta \eta_{xxt}
                                                          + \ve \frac{1}{3} \eta_x \eta_{xt}
                          + \ve \frac{15}{4} \eta^2 \eta_x 
                          + \ve \frac{1}{2} \eta_x \eta_{xx} + \ve \frac{1}{2} \eta \eta_{xxx} \\
=                          \ve \frac{3}{4} \eta^2 \eta_t 
                          + \ve \frac{1}{6} \eta_t \eta_{xx} + \ve \frac{1}{6} \eta \eta_{xxt}
                                                          + \ve \frac{1}{3} \eta_x \eta_{xt}
                          + \ve \frac{3}{4} \eta^2 \eta_x 
                          + \ve \frac{1}{2} \eta_x \eta_{xx} + \ve \frac{1}{6} \eta \eta_{xxx} \\
=                          \ve \frac{3}{4} \eta^2 ( \eta_t + \eta_x) 
                         + \ve \frac{1}{6} \eta \dx^2 ( \eta_t + \eta_x) 
                         + \ve \frac{1}{6} \eta_t \eta_{xx}
                         + \ve \frac{1}{3} \eta_x \eta_{xt}
                         + \ve \frac{1}{2} \eta_x \eta_{xx} \\
= \scrO(\ve^2) + \scrO(\ve^2) + \ve \frac{1}{2} (\eta_t + \eta_x) \eta_{xx}
                             + \ve \frac{1}{3} ( - \eta_t \eta_{xx} + \eta_x \dx \eta_{t} )\\
= \scrO(\ve^2) + \scrO(\ve^2) +  \scrO(\ve^2)
+ \ve \frac{1}{3} \Big( \big( \eta_x + \scrO(\ve) \big) \eta_{xx} 
                                     + \eta_x \dx \big( - \eta_{x} + \scrO(\ve) \big) \Big).
\end{multline*}
The proof now proceeds along the same lines as above. 
\end{proof}
It should be noted that also in the case of the energy balance, the 
$L^2$ norm used in the proofs can be replaced by any Sobolev norm $H^k$ 
as long as the initial data are regular enough.

\section{Discussion}
%
%
In the present letter, it was shown that momentum and energy
conservation hold {\it approximately} in the context of the
KdV equation. The approximate momentum and energy balances
can be made rigorous solely by using well-posedness results
for the KdV equation such as provided in \cite{BonaSmith} and
in many other contributions. 

One interesting aspect of these approximate balance laws is 
that they hold independently of the fidelity of the KdV solution
as an approximation of a solution of the water-wave problem based on the full Euler equations. 
To explain this further, note that it was shown in
\cite{IsrawiKalischCRASS} that the momentum density $I$ defined
above in \eqref{momDensity} approximates the corresponding quantity in
the full Euler equations as long as the solution of the KdV
equation is a close approximation of the full Euler equations.
Indeed it was possible to show that the estimate
\begin{equation}\label{MomApprox}
\Big\| \int_{0}^{1+\ve\eta} \phi_x \, dz - I\Big\|_{H^s} \leq C \ve^2 (1+t)
\end{equation}
holds. Similar estimates can probably be shown for the quantities
$q_I$, $E$ and $q_E$.
Note however, that the estimate \eqref{MomApprox}
degenerates as $t$ gets larger and approaches $t \sim  1 / \ve$.
The results in Theorems 1, 2 and 3 in the present note have no such restriction.
Indeed, they hold {\it globally} for all $t \ge 0$.

Finally, let us point out that having the correct form of
quantities such as $I$, $q_I$, $E$ and $q_E$ can be useful 
when applying the KdV equation in situations where
knowledge of the free surface profile is insufficient.
Indeed, there are situations where the internal dynamics of the flow are 
an important factor.
An example of  such an application is the study of the energy balance in undular bores \cite{AK1}.
The energy flux $q_E$ was also used in a decisive way in a study of nonlinear shoaling in \cite{KK}. \\

\vskip 0.05in
\noindent
{\bf Acknowledgments.}
This research was supported by the Research Council of Norway under grant no. 239033/F20.
HK would like to thank the Isaac Newton Institute for Mathematical Sciences, 
Cambridge University and the Simons foundation for support and hospitality 
when work on this paper was undertaken. 
SI would like to thank University of Bergen and the Department of Mathematics
for support and hospitality when work on this paper was undertaken.


\end{document}